\documentclass[12pt,a4paper]{article}
\usepackage[utf8]{inputenc}

\usepackage{amsmath,amssymb}
\usepackage{amsthm}
\usepackage{enumitem}
\usepackage{appendix}
\usepackage{theoremref}

\usepackage{graphicx}

\usepackage[usenames]{color}

\usepackage{txfonts}

\usepackage[margin=1in]{geometry}
\usepackage{amsmath}
\usepackage{mathrsfs}
\usepackage{amsfonts}
\usepackage{cite}
\usepackage[mathscr]{euscript}
\usepackage{amssymb}
\usepackage{appendix}
\usepackage{tikz-cd}
\usepackage{empheq}
\usepackage{mathtools}
\usepackage{theoremref}
\usepackage{hyperref}
\usepackage{comment}
\usepackage{mathabx}
\usepackage{dsfont}
\usepackage{multicol}
\usepackage{graphicx}
\usepackage{float}
\usepackage{enumitem}
\usepackage{subcaption}
\usepackage{array}
\usepackage{mhsetup}
\usepackage[misc]{ifsym}
\usepackage{datetime}
\usepackage[in]{fullpage}
\usepackage{algorithm}
\usepackage{algorithmic}
\usepackage{authblk}

\newcommand{\ben}{\begin{eqnarray*}}
\newcommand{\enn}{\end{eqnarray*}}

\newtheorem{thm}{Theorem}[section]

\newtheorem{defn}[thm]{Definition}

\newtheorem{prop}[thm]{Proposition}
\newtheorem{cor}[thm]{Corollary}
\newtheorem{rem}[thm]{Remark}

\definecolor{dhcol}{rgb}{0,0.5,0}

\definecolor{cecol}{rgb}{0,0,0.5}

\newcommand{\mres}{%
	\,\raisebox{-.127ex}{\reflectbox{\rotatebox[origin=br]{-90}{$\lnot$}}}\,%
}

\title{Weyl Group Representation of Billiard Trajectories for One-dimensional Hard Sphere Dynamics}
\author{Mark Wilkinson\footnote{Department of Physics and Mathematics, New Hall Block, Nottingham Trent University, Nottingham, United Kingdom (\href{mark.wilkinson@ntu.ac.uk}{mark.wilkinson@ntu.ac.uk}).}}

\date{}

\begin{document}
\maketitle
\begin{abstract}
\noindent We present an exact formula for the dynamics of $N$ hard spheres of radius $r>0$ on an infinite line which evolve under the assumption that total linear momentum and kinetic energy of the system is conserved for all times. This model is commonly known as the one-dimensional \emph{Tonks gas} or the \emph{hard rod gas} model. Our exact formula is expressed as a sum over the Weyl group associated to the root system $A_{N-1}$ and is valid for all initial data in a full-measure subset of the tangent bundle of the hard sphere table. As an application of our explicit formula, we produce a simple proof that the associated billiard flow admits the Liouville measure on the tangent bundle of the hard sphere table as an invariant measure.
\end{abstract}
\section{Introduction}
The study of models of one-dimensional particle systems in both classical and quantum mechanics continues to be an area of active enquiry (cf. \cite{granet2023wavelet, schiltz2023kinetic, gurin2022anomalous, peacock2022quantum, gietka2023squeezing, zhang2022free, saxberg2022disorder}), in part due to the relative simplicity of one-dimensional models compared with their higher-dimensional analogues, but also due to their richness. It is common that one-dimensional models of particle systems are exactly solvable in some sense. We invite the reader to consult the monographs of Baxter \cite{baxter2016exactly} or Franchini \cite{franchini2017introduction} for examples. As regards the study of \emph{classical hard sphere systems} with which the present work is concerned, the well-known work of Tonks \cite{tonks1936complete} established that certain aspects of these models are exactly solvable. Indeed, Tonks was able to derive the exact equation of state for a gas of $N$ spheres whose centres of mass are constrained to lie on a finite line segment. Since the work of Tonks in the 1930s, various other aspects of the one-dimensional hard rod gas model have been shown to be exactly solvable. These results pertain to the behaviour of the hard rod gas over kinetic or hydrodynamic space- and time-scales and for which total momentum and energy of the system is conserved for all time. For instance, it has been shown in the work of Lebowitz, Percus, and Skyes \cite{lebowitz1967kinetic, lebowitz1968time, percus1969exact}, as well as Jepsen \cite{jepsen1965dynamics}, that kinetic equations derived from the underlying hard rod dynamics admit exact solutions. It has also been shown in the work of Boldrighini, Dobrushin, and Sukhov \cite{boldrighini1997one, boldrighini1983one, boldrighini1984hydrodynamics} that it is possible to derive closed form Euler- and Navier-Stokes-like partial differential equations in so-called hydrodynamic limits of the hard rod gas. The solutions of the initial-value problems associated to these PDE are, however, not known to be exact in general.

None of the aforementioned works considers the explicit construction of the underlying momentum- and energy-conserving billiard dynamics governing $N$ hard rods on a line. Owing to the fact that the simultaneous collision of $M\geq 3$ hard rods may be resolved in more than one way whilst conserving total momentum and energy of the system (see Wilkinson \cite{wilkinson2023maximal}), unlike the case of binary collisions, these dynamics are typically only constructed for initial centres of mass and initial velocities outside a non-empty subset of phase space of null measure. This construction has been considered by Sinai \cite{sinai1972ergodic} in the case of infinitely-many hard rods on a line, as well as implicitly by Alexander \cite{alexander1975infinite} in his 1975 Berkeley thesis in both the case of finitely-many and the case of infinitely-many rods. Murphy \cite{murphy1994dynamics} has also considered the case of finitely-many hard rods on a line whose masses and lengths are arbitrary. The construction of dynamics in these works is typically achieved by event-driven algorithms, as opposed to the identification of explicit analytical formulae for the dynamics globally in time. Qualitative properties of the dynamics, in particular ergodicity, have also been studied by Sinai \cite{sinai1972ergodic} as well as by Aizenman, Goldstein and Lebowitz \cite{aizenman1975ergodic}.

In the present article, we show that for any $N\geq 2$, the momentum- and energy-conserving dynamics of $N$ hard rods on an infinite line -- the model on which the statistical mechanics, kinetic theory and hydrodynamics of the hard rod gas are all ultimately based -- is itself exactly solvable in a certain sense. We model the set of admissible centres of the $N$ hard spheres of radius $r>0$ in the system by the so-called \emph{hard sphere table} $\mathcal{P}_{r}\subset\mathbb{R}^{N}$ defined by
\begin{equation}
\mathcal{P}_{r}:=\left\{
X=(x_{1}, ..., x_{N})\in\mathbb{R}^{N}\,:\,|x_{i}-x_{j}|\geq 2r\hspace{2mm}\text{if}\hspace{2mm}i\neq j
\right\},
\end{equation}
and in particular denoting the connected component of $\mathcal{P}_{r}$ consisting of those $X$ whose components are in increasing order by $\mathcal{Q}_{r}$, namely
\begin{equation}
\mathcal{Q}_{r}:=\left\{
X=(x_{1}, ..., x_{N})\in\mathbb{R}^{N}\,:\,x_{1}-2r\leq x_{2}-4r\leq ...\leq x_{N-1}-2(N-1)r\leq x_{N}-2Nr
\right\}.
\end{equation}
In what follows, we prove that for any initial datum $Z_{0}=(X_{0}, V_{0})$ leading to only isolated binary collisions (the set of which is of full measure in the tangent bundle of the hard sphere table), the corresponding piecewise linear trajectory $X:\mathbb{R}\rightarrow\mathbb{R}^{N}$ governing the centres of mass of all $N$ hard rods is given explicitly, up to conjugation by a shift operator and a finite set of times, by
\begin{equation}\label{makeref}
X(t)=\sum_{g\in W}\mathds{1}_{g^{-1}c}(X_{0}+tV_{0})g(X_{0}+tV_{0}),
\end{equation}
where $W$ is the Weyl group of the root system $A_{N-1}$ (cf. Hall \cite{hall2013lie}, Chapter 8), $c$ is a fundamental Weyl chamber associated to $A_{N-1}$, and $\mathds{1}_{g^{-1}c}:\mathbb{R}^{N}\rightarrow\{0, 1\}$ denotes the indicator function of the set $g^{-1}c$. We also prove in turn that the so-called \emph{Liouville measure}, namely the restriction of the $2N$-dimensional Lebesgue measure to phase space, is an invariant measure of the billiard flow $\{T^{t}\}_{t\in\mathbb{R}}$ associated to the billiard trajectory \eqref{makeref} above: see Chernov and Markarian \cite{chernov2006chaotic}. This allows one also to prove that the canonical ensemble is an invariant measure of the billiard flow $\{T^{t}\}_{t\in\mathbb{R}}$, as one would expect. 
\subsection{Statements of Main Results}
Our first main result is a novel representation formula for the dynamics of $N$ hard spheres on an infinite line which conserves the total linear momentum and energy of the system. More precisely, we show that for any initial datum in a certain full-measure subset of phase space the associated billiard trajectory on $T\mathcal{P}_{r}$ can be written as a sum over the Weyl group of the root system $A_{N-1}$. We refer the reader to Section \ref{defsbas} below for definitions and notation. In what follows, we state our results for initial data $X_{0}=(x_{0, 1}, ..., x_{0, N})$ that admit the order $x_{0, 1}-2r\leq x_{0, 2}-4r\leq ...\leq x_{0, N-1}-2(N-1)r\leq x_{0, N}-2Nr$. However, it is possible to modify our analytical formula so that it holds for \emph{any} $X_{0}\in\mathcal{P}_{r}$ through the action of the symmetric group on $N$ letters $\mathfrak{s}(N)$, as we discuss below.
\begin{thm}[Weyl Group Representation Formula for Billiard Dynamics on the Hard Sphere Table]\label{mainres}
Let $N\geq 2$ and $r>0$. There exists a Lebesgue full-measure subset $\mathcal{G}_{r}\subset T\mathcal{Q}_{r}\subset\mathbb{R}^{2N}$ of the tangent bundle of the component of the  hard sphere table $\mathcal{Q}_{r}\subset\mathbb{R}^{N}$ such that for any $Z_{0}=(X_{0}, V_{0})\in\mathcal{G}_{r}$ the unique momentum- and energy-conserving billiard trajectory $X:\mathbb{R}\rightarrow\mathcal{Q}_{r}$ with the properties $X(0)=X_{0}$ and $\frac{d}{dt_{-}}|_{t=0}X=V_{0}$ is given by the formula
\begin{equation}\label{repontableP}
X(t)=\sum_{g\in W}\mathds{1}_{g^{-1}c}(S_{-r}(X_{0})+tV_{0})S_{r}(g(S_{-r}(X_{0})+tV_{0}))    
\end{equation}
for all $t\in\mathbb{R}\setminus\mathcal{T}(Z_{0})$, where $W\subset\mathrm{O}(N)$ is the Weyl group associated to the root system $A_{N-1}$, $c\subset\mathbb{R}^{N}$ is the fundamental chamber associated to the subset of simple roots 
\begin{equation}
\Delta:=\{-e_{i}^{N}+e_{j}^{N}\,:\,1\leq i<j\leq N\}\subseteq A_{N-1},
\end{equation}
$S_{r}:\mathbb{R}^{N}\rightarrow\mathbb{R}^{N}$ is the shift operator defined pointwise by $S_{r}(Y):=Y+r\sum_{i=1}^{N}ie_{i}^{N}$ for all $Y\in\mathbb{R}^{N}$ with $e^{N}_{i}$ denoting the $i$\textsuperscript{th} canonical basis vector in $\mathbb{R}^{N}$, and $\mathcal{T}(Z_{0})\subset\mathbb{R}$ is a finite set of collision times. Moreover, $\mathcal{T}(Z_{0})$ is a set of removable discontinuities of the map \eqref{repontableP}.
\end{thm}
We use the formula \eqref{repontableP} to plot centre-of-mass trajectories in figure \ref{figgy} below. We prove this theorem by working equivalently on what we term the \emph{fundamental table} $\mathcal{Q}$ (whose name derives from the analogous notion of \emph{fundamental chamber} from root system theory). Without the role of the radius $r>0$ in the dynamics of the system, momentum- and energy-conserving billiard trajectories $X:\mathbb{R}\rightarrow\mathcal{Q}$ may be written more simply as    
\begin{equation}\label{simple}
X(t)=\sum_{g\in W}\mathds{1}_{g^{-1}c}(L(t))gL(t),     
\end{equation}
where $L(t):=X_{0}+tV_{0}$ for $t\in\mathbb{R}$. For general data  $Z_{0}=(X_{0}, V_{0})$ for which $X_{0}$ lies in a component of the hard sphere table $\mathcal{P}_{r}$ other than $\mathcal{Q}_{r}$, it can also be shown readily as a consequence of the above result that the corresponding billiard trajectory $X:\mathbb{R}\rightarrow\mathcal{P}_{r}$ is given explicitly by
\begin{equation}
X(t)=\sum_{\pi\in\mathfrak{s}(N)}\sum_{g\in W}\mathds{1}_{T(\pi\mathcal{Q})}(Z_{0})\mathds{1}_{g^{-1}c}(S_{-r}(\pi^{-1}X_{0})+t\pi^{-1}V_{0})S_{r}(\pi g(S_{-r}(\pi^{-1}X_{0})+t\pi^{-1}V_{0})),      
\end{equation}
where $\mathfrak{s}(N)$ denotes the symmetric group on $N$ letters.

Our second main result is a short proof of the fact that the so-called \emph{Liouville measure} $\Lambda_{r}:=\mathscr{L}\mres T\mathcal{P}_{r}$ (where $\mathscr{L}$ denotes the Lebesgue measure on $\mathbb{R}^{2N}$ and $\mres$ the restriction operator) is an invariant measure of any momentum- and energy-conserving billiard flow on $T\mathcal{P}_{r}$.
\begin{thm}[The Liouville Measure is an Invariant Measure of any Momentum- and Energy-Conserving Billiard Flow on $T\mathcal{P}_{r}$]\label{liouv}
Let $N\geq 2$ and $r>0$. Any momentum- and energy-conserving billiard flow $\{T^{t}\}_{t\in\mathbb{R}}$ on $T\mathcal{P}_{r}$ admits the property
\begin{equation}
T^{t}\#\Lambda_{r}=\Lambda_{r}    
\end{equation}
for all $t\in\mathbb{R}$, where $T^{t}\#\Lambda_{r}(E):=\Lambda_{r}(T^{-t}(E))$ for measurable subsets $E\subseteq T\mathcal{P}_{r}$.
\end{thm}
Whilst it is typically expected in the theory of billiards that energy-conserving billiard flows on the tangent bundle of the table admit the analogous Liouville measure as an invariant measure, we believe this is the first time this has been proved rigorously in the context of one-dimensional hard sphere dynamics. Theorem \ref{liouv} admits applications to the rigorous derivation of the weak form of the Liouville equation, as well as the associated BBGKY hierarchy, for example, the likes of which we discuss in the final Section of the paper.
\subsection{Structure of Paper}
In Section \ref{defsbas}, we set out the basic objects with which we work throughout this article. In Section \ref{weylrep}, we prove the claimed representation formula for the dynamics of $N$ hard rods on an infinite line. In Section \ref{louy}, we make use of the previous section to demonstrate that the Liouville measure is an invariant measure of any momentum- and energy-conserving billiard flow on the tangent bundle of the hard rod table. Finally, in Section \ref{closrem}, we close with some remarks on the applications of the results obtained in this article. 
\section{Definitions and Basic Results}\label{defsbas}
In this Section, we lay out the basic definitions of those objects we employ in this article. In particular, we introduce some of the basic objects and associated useful results that pertain to root systems and their Weyl algebras of which we make use in the sequel.
\subsection{The Hard Sphere Billiard Table}
Suppose $r>0$. We write $\mathcal{P}_{r}\subset\mathbb{R}^{N}$ to denote the \textbf{hard sphere billiard table} defined by
\begin{equation}
\mathcal{P}_{r}:=\left\{
X=(x_{1}, ..., x_{N})\in\mathbb{R}^{N}\,:\,|x_{i}-x_{j}|\geq 2r\hspace{2mm}\text{for}\hspace{2mm}i\neq j
\right\}.
\end{equation}
We write $\mathcal{Q}_{r}\subset\mathcal{P}_{r}$ to denote the component of the hard sphere table whose elements lie in ascending order, namely
\begin{equation}
\mathcal{Q}_{r}:=\left\{
X=(x_{1}, ..., x_{N})\in\mathbb{R}^{N}\,:\,x_{1}-2r\leq x_{2}-4r\leq ...\leq x_{N-1}-2(N-1)r\leq x_{N}-2Nr
\right\}.
\end{equation}
We also write $\mathcal{Q}\subset\mathbb{R}^{N}$ to denote the \textbf{fundamental table} defined by
\begin{equation}
\mathcal{Q}:=\left\{
X=(x_{1}, ..., x_{N})\in\mathbb{R}^{N}\,:\,x_{i}\leq x_{i+1}\hspace{2mm}\text{for}\hspace{2mm}i\in\{1, ..., N-1\}
\right\}.
\end{equation}
If we define the \emph{shift map} $S_{r}:\mathbb{R}^{N}\rightarrow\mathbb{R}^{N}$ pointwise by
\begin{equation}
S_{r}(X):=X+r\sum_{i=1}^{N}ie_{i}^{N}    
\end{equation}
for $X\in\mathcal{Q}$, it holds that the hard sphere billiard table $\mathcal{P}_{r}$ admits the representation in terms of the fundamental table $\mathcal{Q}$ given by
\begin{equation}\label{qintermsofp}
\mathcal{P}_{r}=\bigcup_{\pi\in\mathfrak{s}(N)}(S_{r}\circ\pi) \mathcal{Q}.   
\end{equation}
We note that for $N\geq 3$ the fundamental table $\mathcal{Q}\subset\mathbb{R}^{N}$ -- and in turn $\mathcal{P}_{r}\subset\mathbb{R}^{N}$ -- admits the structure of a manifold with corners: see Joyce \cite{joyce2009manifolds}. The boundary $\partial\mathcal{Q}\subset\mathcal{Q}$ is defined to be $\partial\mathcal{Q}:=\mathcal{Q}\setminus\mathcal{Q}^{\circ}$, where $\mathcal{Q}^{\circ}$ denotes the interior of $\mathcal{Q}$ with respect to the Euclidean topology. Owing to the observation \eqref{qintermsofp}, we focus our subsequent attention on the fundamental table $\mathcal{Q}$ and thereby only on one connected component of $\mathcal{P}_{r}$, namely $\mathcal{Q}_{r}$.
\subsection{The Tangent Bundle of the Table}\label{tangentbundle}
We shall work on the tangent bundle $T\mathcal{Q}$ of the table $\mathcal{Q}$ as a means of modelling both the centres of mass $X=(x_{1}, ..., x_{N})\in\mathcal{Q}$ of the $N$ hard spheres in the system as well as the velocities $V=(v_{1}, ..., v_{N})\in\mathbb{R}^{N}$ of their centres. Our main result does not hold for all initial data for billiard trajectories taken in $T\mathcal{Q}$, rather for only a full-measure subset thereof which lead to dynamics involving only isolated binary collisions. For this reason, we require some terminology to describe the set of all initial data for which Theorem \ref{mainres} holds. 
\subsubsection{Links and Chains}
To define the tangent bundle of $\mathcal{Q}$, we require some terminology. Moving forward, we term a nonempty subset of integers $\gamma\subseteq\{1, ..., N\}$ a \emph{link} if and only if it is of the form of a sequence of consecutive integers, i.e. $\gamma=\{i, i+1, ..., i+j-1, i+j\}$ for some $j>0$. The cardinality of a link is denoted by $|\gamma|$ and termed its \emph{length}. Moreover, we call a nonempty subset of integers $\gamma\subseteq\{1, ..., N\}$ a \emph{chain} if and only if it is the disjoint union of any number of links. (In particular, all links are chains but not all chains are links.) We shall abuse notation and write a given chain $\gamma$ in the form $(\gamma_{1}, ..., \gamma_{M})$, where each $\gamma_{i}\subseteq\{1, ..., N\}$ denotes a distinct link comprising the chain. Whenever $N\geq 2$ is understood, we denote the collection of all chains by $\Gamma$. Furthermore, we denote the subset of $\Gamma$ consisting of all chains of length 2 by $\Gamma_{2}$.
\subsubsection{The Tangent Bundle of the Fundamental Table}
With this in place, we define the tangent bundle $T\mathcal{Q}$ of the fundamental table $\mathcal{Q}$ by
\begin{equation}
T\mathcal{Q}:=\bigsqcup_{X\in\mathcal{Q}}T_{X}\mathcal{Q},    
\end{equation}
where the fibres $T_{X}\mathcal{Q}\subset\mathbb{R}^{N}$ are defined for each base point $X\in\mathcal{Q}$ by
\begin{equation}\label{fibres}
T_{X}\mathcal{Q}:=\left\{
\begin{array}{ll}
\mathbb{R}^{N} & \quad \text{if}\hspace{2mm}X\in\mathcal{Q}^{\circ},\vspace{2mm}\\
\mathcal{V}_{\gamma} & \quad \text{if}\hspace{2mm}X\in\partial_{\gamma}\mathcal{Q}\hspace{2mm}\text{for some}\hspace{1mm}\gamma\in\Gamma,
\end{array}
\right.
\end{equation}
where the boundary component $\partial_{\gamma}\mathcal{Q}\subset\mathcal{Q}$ is defined by
\begin{equation}
\partial_{\gamma}\mathcal{Q}:=\bigcap_{k=1}^{M}\left\{
X=(x_{1}, ..., x_{N})\in\partial\mathcal{Q}\,:\,x_{i}=x_{j}\hspace{2mm}\text{for all}\hspace{2mm}i, j\in\gamma_{k}
\right\}
\end{equation}
if $\gamma=(\gamma_{1}, ..., \gamma_{M})$ for some $M\geq 1$, and $\mathcal{V}_{\gamma}\subset\mathbb{R}^{N}$ denotes the set of velocities given by
\begin{equation}
\mathcal{V}_{\gamma}:=\bigcap_{k=1}^{M}\left\{
V=(v_{1}, ..., v_{N})\in\mathbb{R}^{N}\,:\,V\cdot (e_{i}^{N}-e_{i+1}^{N})\geq 0\hspace{2mm}\text{for}\hspace{2mm}\min\gamma_{k}\leq i\leq \min\gamma_{k}+|\gamma_{k}|-2
\right\}
\end{equation}
if $\gamma=(\gamma_{1}, ..., \gamma_{M})$ for some $M\geq 1$. We denote the component of the boundary $\partial\mathcal{Q}$ characterising isolated binary collisions by $\partial_{2}\mathcal{Q}$, namely
\begin{equation}
\partial_{2}\mathcal{Q}:=\bigcup_{\gamma\in\Gamma_{2}}\partial_{\gamma}\mathcal{Q}.    
\end{equation}
We note that $T\mathcal{Q}$ is a Lebesgue-measurable subset of $\mathbb{R}^{2N}$. In contrast with some authors (cf. Joyce \cite{joyce2009manifolds}, Lee \cite{lee2012smooth}), we do not define the fibres $T_{X}\mathcal{Q}$ to be $\mathbb{R}^{N}$ for \emph{all} $X\in\mathcal{Q}$. The reader may readily check in our case that the fibres $T_{X}\mathcal{Q}$ defined in \eqref{fibres} above for $X\in\partial_{\gamma}\mathcal{Q}$ are precisely the sets of all \emph{pre-collisional} velocities of billiard trajectories, namely
\begin{equation}
\left\{
\frac{d}{ds\_}\xi(s)\bigg|_{s=0}\,:\,\xi\in C^{1}_{\pm}(\mathbb{R}, \mathcal{Q})\hspace{2mm}\text{with}\hspace{2mm}\xi(0)\in\partial_{\gamma}\mathcal{Q}
\right\},
\end{equation}
where $C^{1}_{+}(\mathbb{R}, \mathcal{Q})$ and $C^{1}_{-}(\mathbb{R}, \mathcal{Q})$ denote the sets of right- and left-differentiable $\mathcal{Q}$-valued maps on $\mathbb{R}$, respectively. This definition is consistent with the `usual' definition of tangent space at an interior point of a manifold of class $C^{1}$, namely
\begin{equation}
\mathbb{R}^{N}=\left\{
\frac{d}{ds}\xi(s)\bigg|_{s=0}\,:\,\xi\in C^{1}_{\pm}(\mathbb{R}, \mathcal{Q})\hspace{2mm}\text{with}\hspace{2mm}\xi(0)\in\mathcal{Q}^{\circ}
\right\}.   
\end{equation}
We adopt this definition of tangent space as it is necessary if one wishes to define billiard flow operators $T^{t}$ on $T\mathcal{Q}$ whose values $T^{t}(Z_{0})$ are unambiguous if $Z_{0}\in T\mathcal{Q}$ is such that $X_{0}\in\partial\mathcal{Q}$.

In the sequel, we write $\mathscr{L}$ to denote the Lebesgue measure on $\mathbb{R}^{2N}$ and we write $\Lambda$ to denote the restriction of $\mathscr{L}$ to $T\mathcal{Q}$. The measure $\Lambda$ shall be termed subsequently as the \textbf{Liouville measure} on $T\mathcal{Q}$.
\subsection{Good Initial Data in $T\mathcal{Q}$}
The explicit formula for the dynamics of $N$ hard spheres on an infinite line we study herein only holds true for those initial data which give rise to isolated binary collisions, i.e. it does not hold for those data leading to more than 2 spheres simultaneously in contact for some instant of time. In the language of Section \ref{tangentbundle} above, it holds only for those $Z_{0}=(X_{0}, V_{0})\in T\mathcal{Q}$ for which the associated billiard trajectory $X:\mathbb{R}\rightarrow\mathcal{Q}$ has the property 
\begin{equation}
X(t)\in\mathcal{Q}\setminus\left(\bigcup_{\gamma\in\Gamma\setminus\Gamma_{2}}\partial_{\gamma}\mathcal{Q}\right)    
\end{equation}
for all times $t\in\mathbb{R}$. For this reason, we define what we mean by a set of \emph{good} initial data $\mathcal{G}\subseteq T\mathcal{Q}$ leading to (at most) isolated binary collisions, as well as a set of \emph{bad} initial data $\mathcal{B}$ which is the complement of $\mathcal{G}$ in $T\mathcal{Q}$.
\begin{defn}[Good Initial Data in $T\mathcal{Q}$]
Let $N\geq 2$. We write $\mathcal{G}\subset T\mathcal{Q}$ to denote the set of all \textbf{good} initial data in the bundle $T\mathcal{Q}$ given by
\begin{equation}
\mathcal{G}:=\bigsqcup_{X\in\mathcal{Q}}\mathcal{G}_{X},    
\end{equation}
where the fibres $\mathcal{G}_{X}\subset\mathbb{R}^{N}$ are defined for each base point $X\in\mathcal{Q}$ by
\begin{equation}
\mathcal{G}_{X}:=\left\{
V\in\mathbb{R}^{N}\,:\,\ell(X, V)\cap\left(\bigcup_{\gamma\in\Gamma\setminus\Gamma_{2}}\partial_{\gamma}\mathcal{Q}\right)=\varnothing
\right\},
\end{equation}
where $\ell(X, V)\subset\mathbb{R}^{N}$ is the line $\ell(X, V):=\{X+sV\,:\,s\in\mathbb{R}\}$.
\end{defn}
It may be checked that $\mathcal{B}:=T\mathcal{Q}\setminus\mathcal{G}$ is a $\mathscr{L}$-null subset of $T\mathcal{Q}$: we refer the reader to the work of Alexander \cite{alexander1975infinite} for details.
\begin{rem}
For details on how to define billiard dynamics on $T\mathcal{Q}$ for those initial data in $\mathcal{B}$ leading to the simultaneous collision of $M$ spheres on a line with $3\leq M\leq N$, we refer the reader to Wilkinson (see \cite{wilkinson2023maximal}, Section 5).
\end{rem}
\subsection{Billiard Trajectories and Billiard Flows on $T\mathcal{Q}$}
Our two main results pertain to both billiard trajectories on $\mathcal{Q}$, as well as to the billiard flows on $T\mathcal{Q}$ defined in terms of these trajectories and their left derivatives on $\mathbb{R}$. Let us now define these basic objects.
\subsubsection{Billiard Trajectories on $\mathcal{Q}$}
We shall say that a map $X:\mathbb{R}\rightarrow\mathcal{Q}$ is a \textbf{billiard trajectory} on $\mathcal{Q}$ if and only if $X$ is a piecewise linear function whose left- and right-derivatives exist everywhere on $\mathbb{R}$, and whose only points of non-differentiability $\tau\in\mathbb{R}$ are those for which $X(\tau)\in\partial\mathcal{Q}$. The \textbf{velocity map} $V:\mathbb{R}\rightarrow\mathbb{R}^{N}$ associated to a billiard trajectory $X$ is defined pointwise by
\begin{equation}
V(t):=\frac{d}{ds_{-}}\bigg|_{s=t}X(s)    
\end{equation}
for all $t\in\mathbb{R}$. We say that a billiard trajectory \emph{conserves momentum and energy} if and only if its velocity map $V$ admits the properties
\begin{equation}\label{linmom}
V(t)\cdot\mathbf{1}=V(0)\cdot\mathbf{1}    
\end{equation}
and
\begin{equation}\label{energy}
|V(t)|^{2}=|V(0)|^{2}    
\end{equation}
for all $t\in\mathbb{R}$, where $\mathbf{1}\in\mathbb{R}^{N}$ denotes the vector of ones given by
\begin{equation}
\mathbf{1}:=\sum_{i=1}^{N}e_{i}^{N}.    
\end{equation}
We apply the above terminology mutatis mutandis for the case of billiard trajectories $X:\mathbb{R}\rightarrow\mathcal{P}_{r}$ on $\mathcal{P}_{r}$. We state without proof the following well-known result which follows from the work of Alexander \cite{alexander1975infinite, alexander1976time}.
\begin{prop}\label{uniqueness}
Let $N\geq 2$ and suppose $Z_{0}=(X_{0}, V_{0})\in\mathcal{G}$. There exists a unique billiard trajectory $X:\mathbb{R}\rightarrow\mathcal{Q}$ that conserves momentum and energy with the property that $X(0)=X_{0}$ and $V(0)=V_{0}$.
\end{prop}
The above result follows essentially from the fact that the semi-algebraic set comprising all solutions $\xi=(\xi_{1}, \xi_{2})\in\mathbb{R}^{2}$ of the system of polynomial (in)equations
\begin{equation}
\left\{
\begin{array}{l}
\xi_{1}+\xi_{2}=w_{1}+w_{2}\vspace{2mm}\\
\xi_{1}^{2}+\xi_{2}^{2}=w_{1}^{2}+w_{2}^{2}\vspace{2mm}\\
\xi_{1}-\xi_{2}\geq 0
\end{array}
\right.
\end{equation}
is a singleton for any given 2-particle velocity vector $(w_{1}, w_{2})\in\mathbb{R}^{2}$ satisfying $w_{1}-w_{2}\leq 0$. Finally, the set of \textbf{collision times} $\mathcal{T}(X)\subset\mathbb{R}$ of a billiard trajectory $X$ is defined by
\begin{equation}
\mathcal{T}(X):=\left\{
s\in\mathbb{R}\,:\,X(s)\in\partial\mathcal{Q}
\right\}.
\end{equation}
If a billiard trajectory $X$ is uniquely determined by the data $X(0)\in\mathcal{Q}$ and $V(0)\in\mathbb{R}^{N}$, we denote $\mathcal{T}(X)$ simply by $\mathcal{T}(Z_{0})$ where $Z_{0}:=(X(0), V(0))\in T\mathcal{Q}$. We note it follows by definition that the velocity map $V$ of a given billiard trajectory $X$ is a lower semi-continuous function on $\mathbb{R}$ whose set of points of discontinuity is exactly $\mathcal{T}(X)$.
\subsubsection{Billiard Flows on $T\mathcal{Q}$}
As we are also interested in transporting measures on $T\mathcal{Q}$ in this article, we also work with \emph{flow maps} associated to billiard trajectories. In what follows, we write $\Pi_{1}:T\mathcal{Q}\rightarrow\mathcal{Q}$ to denote the canonical projection onto the base space $\mathcal{Q}$, and we write $\Pi_{2}:T\mathcal{Q}\rightarrow\mathbb{R}^{N}$ to denote the canonical projection onto $\mathbb{R}^{N}$. We say that a one-parameter family $\{T^{t}\}_{t\in\mathbb{R}}$ of maps $T^{t}:T\mathcal{Q}\rightarrow T\mathcal{Q}$ is a \textbf{billiard flow} on $T\mathcal{Q}$ if and only if for each $Z_{0}\in T\mathcal{Q}$ the map $t\mapsto (\Pi_{1}\circ T^{t})(Z_{0})$ is a billiard trajectory on $\mathcal{Q}$ outside a finite ($Z_{0}$-dependent) subset of $\mathbb{R}$. Moreover, we say that a billiard flow $\{T^{t}\}_{t\in\mathbb{R}}$ \emph{conserves momentum and energy} if and only if for each $Z_{0}\in T\mathcal{Q}$ the map $t\mapsto(\Pi_{2}\circ T^{t})(Z_{0})$ satisfies \eqref{linmom} and \eqref{energy} above outside a finite ($Z_{0}$-dependent) subset of $\mathbb{R}$. 

As Proposition \ref{uniqueness} above makes clear, the definition of a momentum- and energy-conserving billiard flow map $T^{t}:T\mathcal{Q}\rightarrow T\mathcal{Q}$ is unambiguous for those points in $T\mathcal{Q}$ whose associated momentum- and energy-conserving billiard trajectories admit only isolated binary collisions on $\mathbb{R}$. However, owing to the non-uniqueness of extension of momentum- and energy-conserving billiard trajectories $X:\mathbb{R}\rightarrow\mathcal{Q}$ past those times $\tau\in\mathbb{R}$ for which $X(\tau)\in\partial_{\gamma}\mathcal{Q}$ whenever the chain $\gamma$ admits a link of length greater than 2, defining each $T^{t}$ \emph{globally} on $T\mathcal{Q}$ is problematic: see \cite{wilkinson2023maximal}. However, for the purposes of proving Theorem \ref{liouv} above, this proves to be no barrier. Indeed, any such flow operator $S^{t}$ agrees with our representation formula \eqref{simple} on a Lebesgue full-measure subset of $T\mathcal{Q}$, whence the pushforward measures $S^{t}\#\Lambda$ are indistinguishable as measures on $T\mathcal{Q}$, no matter the particular momentum- and energy-conserving flow with which one works.
\subsection{Root Systems and Weyl Groups}
Let us now set out the elements of root systems of use in our representation formula \eqref{repontableP}. Much of what we discuss here can be found in Hall (\cite{hall2013lie}, Chapter 8). We shall work with the particular root system $A_{N-1}\subset\mathbb{R}^{N}$ defined by
\begin{equation}
A_{N-1}:=\left\{
e_{i}^{N}-e_{j}^{N}\,:\,(i, j)\in\{1, ..., N\}^{2}\setminus\delta
\right\},
\end{equation}
where $\delta$ denotes the diagonal set $\{(1, 1), ..., (N, N)\}$. We define a subset of \emph{simple roots} $\Delta\subset A_{N-1}$ by
\begin{equation}
\Delta:=\left\{
-e_{i}^{N}+e_{j}^{N}\,:\,1\leq i<j\leq N
\right\}.
\end{equation}
The collection of \emph{Weyl chambers} associated to the root system $A_{N-1}$ comprises all those connected components of the set 
\begin{equation}
\mathbb{R}^{N}\setminus\left(
\bigcup_{\alpha\in A_{N-1}}\left\{
Y\in\mathbb{R}^{N}\,:\,Y\cdot\alpha=0
\right\}
\right).
\end{equation}
Associated to these Weyl chambers, we also define the boundary hyperplanes $\Pi_{i, j}\subset\mathbb{R}^{N}$ for each $1\leq i<j\leq N$ by
\begin{equation}
\Pi_{i, j}:=\left\{
Y=(y_{1}, ..., y_{N})\in\mathbb{R}^{N}\,:\,y_{i}=y_{j}
\right\}.
\end{equation}
We note that $\Pi_{i, j}=\{Y\in\mathbb{R}^{N}\,:\,Y\cdot\alpha_{i, j}=0\}$, where $\alpha_{i, j}:=-e_{i}^{N}+e_{j}^{N}$. In turn, we define the pencil of hyperplanes $\Pi\subset\mathbb{R}^{N}$ by
\begin{equation}
\Pi:=\bigcup_{i=1}^{N}\bigcup_{j>i}\Pi_{i, j}.    
\end{equation}
The \emph{fundamental Weyl chamber} $c\subset\mathbb{R}^{N}$ (with respect to $\Delta$) is defined by
\begin{equation}
c:=\left\{
Y\in\mathbb{R}^{N}\,:\, Y\cdot\beta>0\hspace{2mm}\text{for all}\hspace{2mm}\beta\in\Delta
\right\}.
\end{equation}
We remark in passing that with the set of simple roots $\Delta$ so-defined it holds that $c=\mathcal{Q}^{\circ}$, the interior of the fundamental table $\mathcal{Q}$. The \emph{Weyl group} $W\subset\mathrm{O}(N)$ of the root system $A_{N-1}$ is defined to be
\begin{equation}
W:=\left\langle
\left\{
I_{N}-(e_{i}^{N}-e_{j}^{N})\otimes(e_{i}^{N}-e_{j}^{N})\,:\,1\leq i<j\leq N
\right\}
\right\rangle,
\end{equation}
where $\otimes$ denotes the tensor product on $\mathbb{R}^{N}\times\mathbb{R}^{N}$, and where the angular brackets $\langle\cdot\rangle$ denote closure with respect to matrix multiplication. Let us now store the following basic result upon which we shall draw in the sequel.
\begin{prop}\label{transitive}
Suppose $X\in\mathbb{R}^{N}\setminus\Pi$. It holds there exists a unique $g\in W$ such that $X\in g^{-1}c$.
\end{prop}
\begin{proof}
See Hall (\cite{hall2013lie}, Chapter 8.7). 
\end{proof}
\section{The Weyl Group Representation Formula}\label{weylrep}
As highlighted above, we choose to work on the fundamental table $\mathcal{Q}$ rather than $\mathcal{P}_{r}$ in order to remove the inessential presence of the rod diameter $2r$. We now prove the following Theorem. 
\begin{thm}\label{repo}
Let $N\geq 2$. For each $Z_{0}\in\mathcal{G}$, the map $\widetilde{X}:\mathbb{R}\rightarrow\mathcal{Q}$ defined pointwise by
\begin{equation}\label{formulaforQ}
\widetilde{X}(t):=\sum_{g\in W}\mathds{1}_{g^{-1}c}(X_{0}+tV_{0})   g(X_{0}+tV_{0}) 
\end{equation}
for $t\in\mathbb{R}$ admits the property that the associated map $X:\mathbb{R}\rightarrow\mathcal{Q}$ defined by 
\begin{equation}\label{fudged}
X(t):=\left\{
\begin{array}{ll}
\widetilde{X}(t) & \quad \text{if}\hspace{2mm}t\in\mathbb{R}\setminus\mathcal{D}(Z_{0}), \vspace{2mm}\\
\displaystyle\lim_{s\rightarrow t}\widetilde{X}(s) & \quad\text{if}\hspace{2mm}t\in\mathcal{D}(Z_{0}),
\end{array}
\right.
\end{equation}
is a billiard trajectory that conserves momentum and energy, where $\mathcal{D}(Z_{0})\subset\mathbb{R}$ is the finite set of (removable) discontinuities of the map $\widetilde{X}$, namely
\begin{equation}\label{discontinuityset}
\mathcal{D}(Z_{0}):=\left\{
s\in\mathbb{R}\,:\,X_{0}+sV_{0}\in\Pi
\right\}.
\end{equation}
\end{thm}
\begin{proof}
We break the proof down into several steps. Firstly, we show that $\widetilde{X}(t)\in\mathcal{Q}$ for all $t\in\mathbb{R}$. Secondly, we identify the finite set of times $\mathcal{D}(Z_{0})\subset\mathbb{R}$ at which $\widetilde{X}$ is discontinuous. Thirdly, in order to show that $X$ as defined in \eqref{fudged} above is a billiard trajectory on $\mathcal{Q}$, we show that $\lim_{s\rightarrow \tau}\widetilde{X}(s)$ exists for any $\tau\in\mathcal{D}(\widetilde{X})$, whence $\widetilde{X}$ admits finitely-many removable points of discontinuity. Finally, we show that the left-derivative of $X$ satisfies the conservation of momentum and the conservation of energy.

Owing to the fact that the closure of the union of all Weyl chambers associated to the root system $A_{N-1}$ is $\mathbb{R}^{N}$ (c.f. Proposition \ref{transitive}), it is manifest that $\widetilde{X}$ in \eqref{formulaforQ} is piecewise linear on $\mathbb{R}\setminus\mathcal{D}(Z_{0})$. However, in order to prove that this formula defines a billiard trajectory on $\mathcal{Q}$, we must demonstrate that $\widetilde{X}$ so-defined admits (i) its range in $\mathcal{Q}$, and (ii) extends continuously from $\mathbb{R}\setminus\mathcal{D}(Z_{0})$ to $\mathbb{R}$. Indeed, recalling that $\mathcal{Q}^{\circ}=c$, for any $\beta\in\Delta$ it holds that
\begin{equation}
\begin{array}{lcl}
\widetilde{X}(t)\cdot\beta & = & \displaystyle \sum_{g\in W}\mathds{1}_{g^{-1}c}(X_{0}+tV_{0})g(X_{0}+tV_{0})\cdot\beta\vspace{2mm}\\
& = & h(X_{0}+tV_{0})\cdot \beta \vspace{2mm}\\
& > & 0
\end{array}
\end{equation}
for all $t\in\mathbb{R}\setminus\mathcal{D}(\widetilde{X})$ by definition of the fundamental Weyl chamber $c$, where $h=h(Z_{0}, t)\in W$ is the unique group element satisfying the identity $X_{0}+tV_{0}\in h^{-1}c$. If $t\in\mathcal{D}(\widetilde{X})$, then $\widetilde{X}(t)=0$, whence $\widetilde{X}(t)\in\mathcal{Q}$. As $Z_{0}\in\mathcal{G}$, the set $\mathcal{D}(Z_{0})$ as defined in \eqref{discontinuityset} above is a finite set of cardinality $N_{\star}:=N(N-1)/2$. It constitutes a set of points at which $\widetilde{X}$ is discontinuous as $\Pi\cap(\cup_{g\in W}g^{-1}c)=\varnothing$, whence $\widetilde{X}(\tau)=0\in\mathbb{R}^{N}$ for all $\tau\in\mathcal{D}(Z_{0})$. By definition of $\mathcal{G}$, it cannot be the case that $X_{0}+sV_{0}=0\in\mathbb{R}^{N}$ for any $s\in\mathbb{R}$, and so $\mathcal{D}(Z_{0})$ characterises the set of points at which $\widetilde{X}$ is discontinuous on $\mathbb{R}$.

It remains to show that the points of discontinuity $\mathcal{D}(Z_{0})$ of the function $\widetilde{X}$ are removable. Let $\tau\in \mathcal{D}(Z_{0})$ be given. We note that $\lim_{t\rightarrow\tau}(X_{0}+tV_{0})\in \Pi_{i,j}$ for some $1\leq i<j\leq N$. As Weyl chambers are open sets, there exist numbers $\eta_{-}, \eta_{+}>0$ and Weyl group elements $g_{-}, g_{+}\in W$ such that $X_{0}+tV_{0}\in g_{-}^{-1}c$ for all $\tau-\eta_{-}<t<\tau$ and $X_{0}+tV_{0}\in g_{+}^{-1}c$ for all $\tau<t<\tau+\eta_{+}$. However, as the Weyl chambers $g_{-}^{-1}c$ and $g_{+}^{-1}c$ are adjacent to one another and are separated by the boundary plane $\Pi_{i, j}$, it holds that $g_{+}^{-1}c=\sigma_{i, j}^{-1}g_{-}^{-1}c$, whence $g_{+}=g_{-}\sigma_{i, j}$ in $W$. Finally, for all $\tau-\eta_{-}<t<\tau$, it holds from formula \eqref{formulaforQ} that
\begin{equation}
X(t)=g_{-}(X_{0}+tV_{0}),    
\end{equation}
whence
\begin{equation}
\lim_{t\uparrow\tau}X(t)=g_{-}(X_{0}+\tau V_{0}).    
\end{equation}
Moreover, for all $\tau<t<\tau+\eta_{+}$, it holds that
\begin{equation}
X(t)=g_{-}\sigma_{i, j}(X_{0}+tV_{0}),    
\end{equation}
whence
\begin{equation}
\lim_{t\downarrow\tau}X(t)=g_{-}\sigma_{i, j}(X_{0}+\tau V_{0})=g_{-}(X_{0}+\tau V_{0}),     
\end{equation}
as $X_{0}+\tau V_{0}\in \Pi_{i, j}$. Consequently, the limit of $\widetilde{X}$ at $\tau\in \mathcal{D}(Z_{0})$ exists and is not equal to $0\in\mathbb{R}^{N}$. Thus, let us define the continuous map $X:\mathbb{R}\rightarrow\mathcal{Q}$ pointwise by
\begin{equation}\label{betterdef}
X(t):=\left\{
\begin{array}{ll}
\widetilde{X}(t) & \quad \text{if}\hspace{2mm}t\in\mathbb{R}\setminus\mathcal{D}(Z_{0}), \vspace{2mm}\\
\lim_{s\rightarrow t}\widetilde{X}(s) & \quad\text{if}\hspace{2mm}t\in\mathcal{D}(Z_{0}).
\end{array}
\right.
\end{equation}
As continuous piecewise linear maps $X:\mathbb{R}\rightarrow\mathcal{Q}$ with finitely-many points of non-differentiability are both left- and right-differentiable everywhere on $\mathbb{R}$, we conclude that $X$ defined by \eqref{betterdef} is a billiard trajectory on $\mathcal{Q}$. In particular, it holds that $\mathcal{D}(\widetilde{X})=\mathcal{T}(X)$.

Now that we have demonstrated that $X$ is a billiard trajectory on $\mathcal{Q}$, its associated velocity map $V$ is well defined. Indeed, it holds by direct computation that
\begin{equation}
V(t)=\frac{d}{ds_{-}}\bigg|_{s=t}X(s)=\sum_{g\in W}\mathds{1}_{g^{-1}c}(X_{0}+tV_{0}) gV_{0}   
\end{equation}
for all $t\in\mathbb{R}\setminus\mathcal{T}(Z_{0})$, and that
\begin{equation}
V(\tau)=\lim_{t\uparrow\tau}V(t)    
\end{equation}
for all $\tau\in\mathcal{T}(X)$. Now, as every element $g\in W\subset\mathrm{O}(N)$ admits $\mathbf{1}\in\mathbb{R}^{N}$ as an eigenvector with corresponding eigenvalue 1, it follows that
\begin{equation}
\begin{array}{lcl}
V(t)\cdot\mathbf{1} & = & \displaystyle \sum_{g\in W}\mathds{1}_{g^{-1}c}(X_{0}+tV_{0})(gV_{0}\cdot\mathbf{1}) \vspace{2mm}\\
& = & \displaystyle\sum_{g\in W}\mathds{1}_{g^{-1}c}(X_{0}+tV_{0})(V_{0}\cdot g^{T}\mathbf{1})\vspace{2mm}\\
& = & \displaystyle\left(\sum_{g\in W}\mathds{1}_{g^{-1}c}(X_{0}+tV_{0})\right)V_{0}\cdot\mathbf{1}\vspace{2mm}\\
& = & V_{0}\cdot \mathbf{1}
\end{array}
\end{equation}
for all $t\in\mathcal{T}(Z_{0})$. Moreover, owing to the facts that the Weyl group comprises only orthogonal matrices and that distinct Weyl chambers are disjoint subsets of $\mathbb{R}^{N}$, it holds that
\begin{equation}
\begin{array}{lcl}
|V(t)|^{2} & = & \displaystyle \left|\sum_{g\in W}\mathds{1}_{g^{-1}c}(X_{0}+tV_{0})gV_{0}\right|^{2} \vspace{2mm}\\
& = & \displaystyle \sum_{g\in W}\mathds{1}_{g^{-1}c}(X_{0}+tV_{0})|gV_{0}|^{2}\vspace{2mm}\\
& = & \displaystyle \left(\sum_{g\in W}\mathds{1}_{g^{-1}c}(X_{0}+tV_{0})\right)|V_{0}|^{2}\vspace{2mm}\\
& = & |V_{0}|^{2}
\end{array}
\end{equation}
for all $t\in\mathbb{R}\setminus\mathcal{T}(Z_{0})$. As a result, $V$ satisfies \eqref{linmom} and \eqref{energy} above.
\end{proof}
\begin{figure}
    \centering
    \includegraphics[angle=-90,origin=c, scale=0.50]{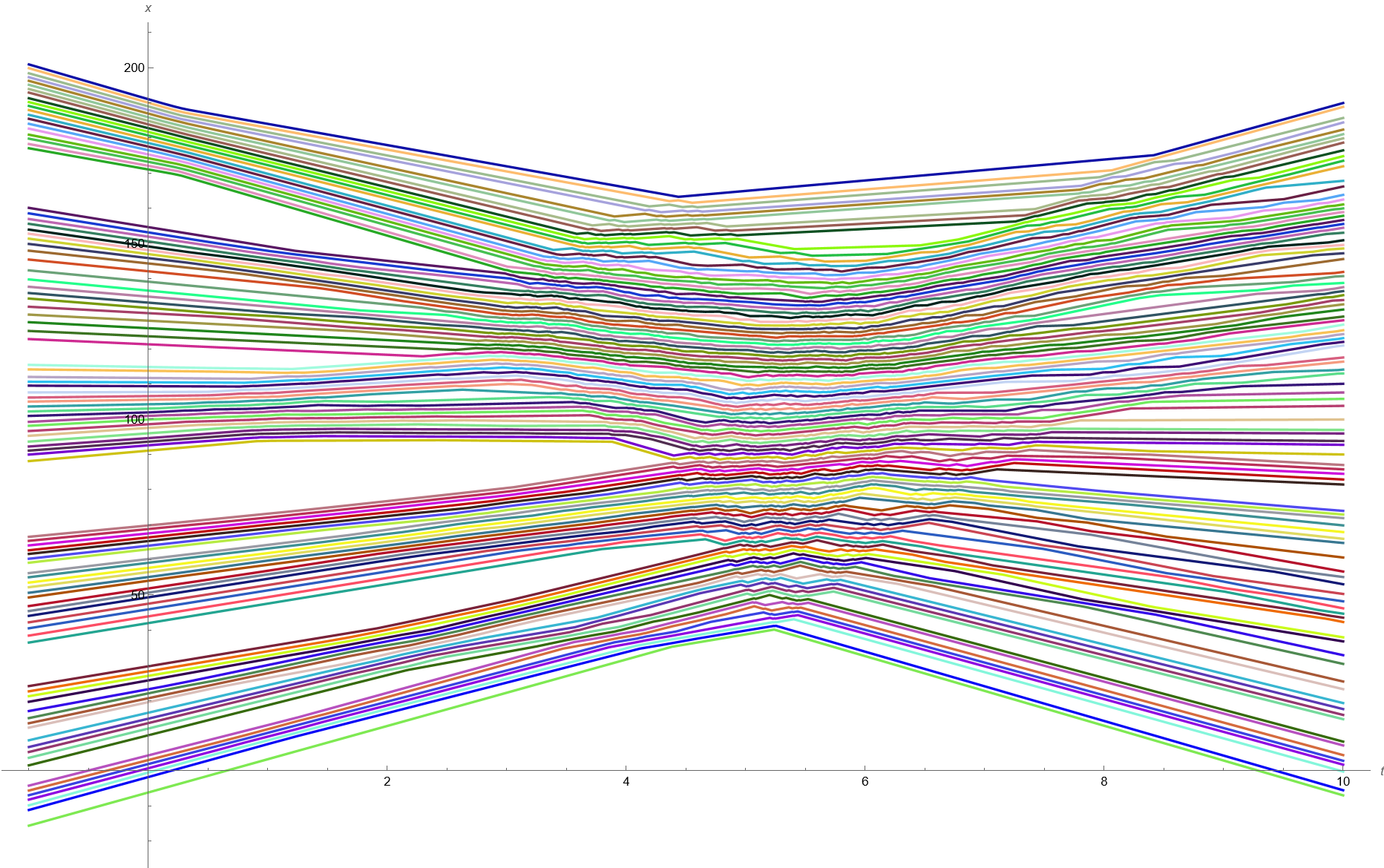}
    \caption{An illustration of the analytical formula \eqref{repontableP} in the case $N=100$. The above diagram shows 100 superimposed centre-of-mass maps $t\mapsto x_{i}(t)$ for $i\in\{1, ..., 100\}$ defined pointwise by $x_{i}(t):=X(t)\cdot e^{100}_{i}$ for $t\in\mathbb{R}$.
    }
    \label{figgy}
\end{figure}
As a simple corollary of the above result on the tangent bundle $T\mathcal{Q}$ of the fundamental table $\mathcal{Q}$, we now prove Theorem \ref{mainres} of this article.
\begin{cor}
Let $N\geq 2$ and $r>0$. For any $Z_{0}=(X_{0}, V_{0})\in \mathcal{G}_{r}:=S_{r}\mathcal{G}$ the unique momentum- and energy-conserving billiard trajectory $X:\mathbb{R}\rightarrow\mathcal{P}_{r}$ with the properties $X(0)=X_{0}$ and $\frac{d}{dt_{-}}|_{t=0}X=V_{0}$ is given by the formula
\begin{equation}
X(t)=\sum_{g\in W}\mathds{1}_{g^{-1}c}(S_{-r}(X_{0})+tV_{0})S_{r}(g(S_{-r}(X_{0})+tV_{0}))    
\end{equation}
for all $t\in\mathbb{R}$,.
\end{cor}
\begin{proof}
This follows by a simple conjugation of the dynamics on the fundamental table by the shift operator $S_{r}$.    
\end{proof}
\section{Invariance of the Liouville Measure}\label{louy}
Now that we have constructed billiard trajectories corresponding to any initial datum in a full-measure subset of $T\mathcal{Q}$, we may use these to define an associated billiard flow $\{T^{t}\}_{t\in\mathbb{R}}$ whose trajectories are momentum- and energy-conserving billiard trajectories on $\mathcal{Q}$ outside a finite set of removable discontinuities. Indeed, for each $t\in\mathbb{R}$, we define the map $T^{t}:T\mathcal{Q}\rightarrow T\mathcal{Q}$ pointwise by
\begin{equation}
T^{t}(Z):=\sum_{g\in W}\mathds{1}_{g^{-1}c}(X+tV)\mathrm{diag}(g, g)\left(I_{2N}+t\sum_{i=1}^{N}e_{i}^{2N}\otimes e_{i+N}^{2N}\right)Z
\end{equation}
for $Z=(X, V)\in T\mathcal{Q}$, where $\mathrm{diag}(g, g)\in\mathbb{R}^{2N\times 2N}$ denotes the matrix given blockwise by
\begin{equation}
\mathrm{diag}(g, g):=\left(
\begin{array}{cc}
g & 0_{N} \\
0_{N} & g
\end{array}
\right)
\end{equation}
with $0_{N}\in\mathbb{R}^{N\times N}$ denoting the zero matrix. By Theorem \ref{repo} above, it holds that for each $Z\in\mathcal{G}$ the map $t\mapsto\Pi_{1}\circ T^{t}(Z)$ is equal, up to a finite set of removable discontinuities, to a momentum- and energy-conserving billiard trajectory. In turn, by definition $\{T^{t}\}_{t\in\mathbb{R}}$ is a billiard flow on $T\mathcal{Q}$ which conserves momentum and energy. Moreover, by the uniqueness of momentum- and energy-conserving trajectories starting at data in $\mathcal{G}$, it holds that for any momentum- and energy-conserving billiard flow $\{S^{t}\}_{t\in\mathbb{R}}$ and all $t\in\mathbb{R}$, the set $\{Z\in T\mathcal{Q}\,:\,S^{t}(Z)=T^{t}(Z)\}$ is of full measure.

Let us now prove the second main Theorem of this work.
\begin{thm}
Let $N\geq 2$. The one-parameter family $\{T^{t}\}_{t\in\mathbb{R}}$ admits the property that
\begin{equation}
T^{t}\#\Lambda=\Lambda    
\end{equation}
for all $t\in\mathbb{R}$, where $\Lambda=\mathscr{L}\mres T\mathcal{Q}$ is the Liouville measure on $T\mathcal{Q}$.
\end{thm}
\begin{proof}
As $T\mathcal{Q}$ endowed with the subspace Euclidean topology on $\mathbb{R}^{2N}$ is a normal topological space, it follows by Urysohn's lemma that for any closed subset $E\subset T\mathcal{Q}$ of finite mass there exists $\{\Phi_{i}\}_{i=1}^{\infty}\subset C^{0}(T\mathcal{Q})\cap L^{1}(T\mathcal{Q})$ such that $\Lambda(E)=\lim_{i\rightarrow\infty}\langle \Lambda, \Phi_{i}\rangle $, where the duality brackets are defined by 
\begin{equation}
\langle \Lambda, \Phi\rangle:=\int_{T\mathcal{Q}}\Phi\,d\Lambda    
\end{equation}
for all $\Phi\in C^{0}(T\mathcal{Q})\cap L^{1}(T\mathcal{Q})$. Now, for any $\Phi\in C^{0}(T\mathcal{Q})\cap L^{1}(\mathcal{Q})$, we find that
\begin{equation}
\begin{array}{cl}
& \langle T^{t}\#\Lambda, \Phi \rangle \vspace{2mm} \\
= & \displaystyle \int_{T\mathcal{Q}}\Phi(T^{t}(Z))\,dZ \vspace{2mm} \\
= & \displaystyle \int_{T\mathcal{Q}}\Phi\left(\sum_{g\in W}\mathds{1}_{g^{-1}c}(X+tV)g(X+tV), \sum_{g\in W}\mathds{1}_{g^{-1}c}(X+tV)gV\right)\,dZ,    
\end{array}
\end{equation}
and since it holds that 
\begin{equation}
\mathds{1}_{g^{-1}c}(X+tV)=\mathds{1}_{(g^{-1}c-tV)\cap c}(X)    
\end{equation}
for each fixed $t\in\mathbb{R}$ and $Z=(X, V)\in T\mathcal{Q}$, we deduce that
\begin{equation}
\begin{array}{cl}
& \langle T^{t}\#\Lambda, \Phi\rangle\vspace{2mm}\\
= & \displaystyle\int_{T\mathcal{Q}}\Phi\left(\sum_{g\in W}\mathds{1}_{(g^{-1}c-tV)\cap c}(X)g(X+tV), \sum_{g\in W}\mathds{1}_{(g^{-1}c-tV)\cap c}(X)gV \right)\,dZ \vspace{2mm}\\
= & \displaystyle\sum_{g\in W}\int_{\mathbb{R}^{N}}\int_{(g^{-1}c-tV)\cap c}\Phi(g(X+tV), gV)\,dXdV\vspace{2mm}\\
= & \displaystyle \sum_{g\in W}\int_{\mathbb{R}^{N}}\int_{g^{-1} c\cap(c+tV)}\Phi(gX, gV)\,dXdV\vspace{2mm}\\
= & \displaystyle \sum_{g\in W}\int_{\mathbb{R}^{N}}\int_{c\cap(g(c+tV))}\Phi(X, V)\,dXdV\vspace{2mm}\\
= & \displaystyle\int_{\mathbb{R}^{N}}\sum_{g\in W}\left(\int_{\mathcal{Q}\cap(g(c+tV))}\Phi(X, V)\,dX\right)dV\vspace{2mm}\\
= & \displaystyle\int_{\mathbb{R}^{N}}\int_{\mathcal{Q}\cap(\cup_{g\in W}g(c+tV))}\Phi(X, V)\,dXdV.
\end{array}
\end{equation}
However, since it holds that
\begin{equation}
\mathbb{R}^{N}\setminus\left(\bigcup_{g\in W}g(c+tV)\right)
\end{equation}
is a null set, we find that
\begin{equation}
\langle T^{t}\#\Lambda, \Phi\rangle=\int_{\mathbb{R}^{N}}\int_{\mathcal{Q}\cap\mathbb{R}^{N}}\Phi(X, V)\,dXdV, 
\end{equation}
whence it holds that
\begin{equation}
\langle T^{t}\#\Lambda, \Phi\rangle=\langle \Lambda, \Phi\rangle    
\end{equation}
for all $\Phi\in C^{0}(T\mathcal{Q})\cap L^{1}(T\mathcal{Q})$. It is now trivial to show that for any measurable $E\subseteq T\mathcal{Q}$ it holds that $T^{t}\#\Lambda(E)=\Lambda(E)$, whence we conclude that $\Lambda$ is an invariant measure of the one-parameter family of flow operators $\{T^{t}\}_{t\in\mathbb{R}}$.
\end{proof}
The proof of Theorem \ref{liouv} is now immediate, as the shift operators $S_{r}$, as well as the action of the permutations $\pi\in\mathfrak{s}(N)$, are measure-preserving for any $r\in\mathbb{R}$. As an immediate consequence of the above result, we note that the (non-normalised) canonical ensemble $\mu:=E_{\beta}\Lambda_{r}$ whose density $E_{\beta}$ is given by
\begin{equation}
E_{\beta}(Z):=\mathrm{exp}\left(-\beta\sum_{i=1}^{N}|v_{i}|^{2}\right)    
\end{equation}
for all $Z=(X, V)\in T\mathcal{P}_{r}$ and all $\beta>0$ is invariant under any momentum- and energy-conserving billiard flow on $T\mathcal{P}_{r}$.
\section{Closing Remarks}\label{closrem}
In this article, we have established an analytical formula for the dynamics of $N$ hard spheres on a line when initial data $Z_{0}=(X_{0}, V_{0})$ lie in a full-measure subset of phase space. We applied this analytical formula to prove that any momentum- and energy-conserving billiard flow $\{S^{t}\}_{t\in\mathbb{R}}$ on the tangent bundle of the hard sphere table admits the Liouville measure as an invariant measure. Our analytical formula may also be employed to construct both mild and weak solutions of the Liouville equation for $N$ hard rods as well as its associated BBGKY hierarchy: see Gallagher, Saint-Raymond and Texier \cite{gallagher2013newton} for further details.

We note that the representation formula \eqref{repontableP} extends formally to the case of \emph{infinitely-many} hard rods on a line, whereby the Weyl group $W$ of the root system $A_{N-1}$ is replaced by a suitable subgroup of bijections of the integers, and the fundamental chamber $c\subset\mathcal{Q}$ is replaced by the subset of $\mathbb{R}^{\omega}$ given by
\begin{equation}
\left\{X=\{x_{i}\}_{i\in\mathbb{Z}}\in\mathbb{R}^{\omega}\,:\,x_{j}\leq x_{j+1}\hspace{2mm}\text{for all}\hspace{1mm}j\in\mathbb{Z}\right\};    
\end{equation}
see Munkres (\cite{munkres2018topology}, Chapter 2). Such a representation formula in this case would offer an alternative proof to that of Sinai \cite{sinai1972ergodic} of the construction of dynamics of infinitely-many hard rods on a line. We do not pursue the details of this extension here. We also believe it possible in the case of finite $N\geq 2$ to extend the representation formula \eqref{repontableP} to \emph{all} initial data in $T\mathcal{P}_{r}$, not simply the good subset of data leading only to binary collisions. The main complications in this extension stem from (i) the structure of the boundaries of the Weyl chambers, in that they are manifolds with corners with components of every possible codimension in $\mathbb{R}^{N-1}$; and (ii) the fact that the momentum- and energy-conserving scattering of $M$ hard rods in simultaneous contact (with $M\geq 3$) is not uniquely defined, and so it is not immediately clear which scattering one ought to adopt in the pursuit of an analytical formula. This constitutes future work. It also would be of interest to investigate the possibility of establishing representation formulae for a class of billiard tables other than the very particular tables $\mathcal{Q}$ and $\mathcal{P}_{r}$ which were under study in this article.
\subsection*{Acknowledgments}
The author would like to thank Bart Vlaar for introducing him to the theory of root systems, and for several illuminating conversations related to the material in this article. The author would also like to thank Pierangelo Marcati for his hospitality at the Gran Sasso Science Institute (GSSI) in L'Aquila, which afforded a wonderful research environment and where a part of this work was done. 
\bibliographystyle{siam}
\bibliography{bibby}
\end{document}